\documentclass[aps,pre,reprint,groupedaddress,showpacs,byrevtex]{revtex4-1}
\usepackage{graphicx}
\usepackage{color}
\usepackage{amsmath}
\usepackage{amsthm,amsfonts,amssymb}
\newtheorem{thm}{Theorem}
\theoremstyle{definition}
\newtheorem{mydef}{Definition}
\begin{document}
\title{Purely competitive evolutionary dynamics for games}
\author{Carl~Veller}
\email{carl.veller@gmail.com}
\affiliation{Department of Mathematics and Applied Mathematics, University of Cape Town, Private Bag X3, Rondebosch 7701, South Africa}
\author{Vinesh~Rajpaul}
\email{vinesh.rajpaul@uct.ac.za}
\affiliation{Astrophysics, Cosmology and Gravity Centre (ACGC), University of Cape Town, Private Bag X3, Rondebosch 7701, South Africa}
\date{\today}
\begin{abstract}
We introduce and analyze a purely competitive dynamics for the evolution of an infinite population subject to a 3-strategy game. We argue that this dynamics represents a characterization of how certain systems, both natural and artificial, are governed. In each period, the population is randomly sorted into pairs, which engage in a once-off play of the game; the probability that a member propagates its type to its offspring is proportional only to its payoff within the pair. We show that if a type is dominant (obtains higher payoffs in games with both other types), its `pure' population state, comprising only members of that type, is globally attracting. If there is no dominant type, there is an unstable `mixed' fixed point; the population state eventually oscillates between the three near-pure states. We then allow for mutations, where offspring have a non-zero probability of randomly changing their type. In this case, the existence of a dominant type renders a point near its pure state globally attracting. If no dominant type exists, a supercritical Hopf bifurcation occurs at the unique mixed fixed point, and above a critical (typically low) mutation rate, this fixed point becomes globally attracting: the implication is that even very low mutation rates can stabilize a system that would, in the absence of mutations, be unstable.
\end{abstract}

\maketitle
\section{Introduction} \label{sec:intro}
A relatively recent development in the game theory literature, that of the dynamic treatment of games, has proven successful in providing answers both to the question of which equilibria will arise in a given game, and how they might be reached; both of these questions have proven problematic in the static, introspective approach to game theory \cite{fudenberg1998}. The fundamental idea behind the dynamic treatment is to consider repetitions of a game, rather than a once-off play, thus allowing players' strategies to \emph{evolve} through time. Rules are posited to govern the evolution of the players' strategies, with such rules being either mindless (occurring exogenous to any decision making of the players -- we call such a rule a `dynamics') or minded (where the rule depends explicitly on players' decisions -- we call such a rule a `learning rule'). Different interpretations of the same game might render the same rule either mindless or minded, so that the distinction is not clearly set.

Our focus in this paper will be on a particular rule, which, as we shall discuss, can be interpreted either as a dynamics or a learning rule. Such rules are often most comfortably cast in the language of (biological) evolution, and it is such a description, for the most part, that we shall employ. Typically, a population of agents is considered, with each agent being one of a number of types (the types correspond to strategies). In each period, a subset of the population is paired up, with each pair engaging in a once-off play of the game (which might be termed a fight, a mating, etc., depending on the general interpretation). The result of the pairings and engagements determines the make-up of the population in the next period, when the process is repeated.

The dynamics thus provides a vehicle by which an outcome of the game (the long term make-up of the population, for example) may be reached. For many sensibly founded dynamics (there are conditions for what constitutes a `sensible' dynamics, typically based on the general principle that relative success of a certain type/strategy should lead to its increase as a proportion of the population -- for more technical definitions, a good reference is \citet{samuelson1998}) this outcome tends to correspond to some of the equilibrium concepts defined for once-off games, although it is not unusual for this not to be the case. That is, a sensible dynamics may result in an equilibrium that is not even Nash (see later in the present paper, as well as \citet{fudenberg1998}). 

It is instructive, in terms of understanding the nature of a typical dynamics and also to provide a foundation and a foil for the model that we develop in the main part of the present paper, briefly to describe what is probably the most well-known and widely-applied dynamics, the `replicator dynamics', first introduced by \citet{taylor1978}. An infinite population comprising $n$ types of agent is assumed; the population state is defined as the proportions of each type in the population: 
\[\mathbf{p} = (p_1,p_2,\ldots,p_n)\in \Delta^n,\]
where $\Delta^n$ is the unit simplex in $\mathbb{R}^n$:
\[ \Delta^n:= \left\{\mathbf{p}\in \mathbb{R}^n: \mathbf{p}\geq \mathbf{0} \; \wedge \; \sum_{i=1}^n p_i = 1\right\}. \]

A type's `fitness' at a given time is defined as the expected payoff to one of its agents in a once-off game with a randomly selected opponent from the population; the proportional growth rate of a given type's proportion at a given time is equal to the difference between that type's fitness and the average fitness throughout the entire population. 

The key concept behind the replicator dynamics is of the fitness of types within the population. This is determined not only by how a certain type fares against other types, but also by how well agents of that type do in engagements with agents of the \emph{same} type. Indeed, this is a vital consideration in the concept of evolutionary stability under the replicator dynamics -- a population state comprising only one type of agent is resistant to invasions by other types if, and only if, the predominant type fares better against itself than the potential invader does against the predominant type \cite{hofbauer1998}. This is clearly a sensible consideration in modeling certain systems; it has been argued that the replicator dynamics provides a model for some aspects of evolution in the natural world \cite{hofbauer2003}, where a type whose agents do very poorly amongst themselves would not be expected to succeed in increasing as a proportion of the population, while a type whose agents do very well amongst themselves would be expected at least to be able to maintain a high proportional representation in the population.

However, in certain circumstances, this `cooperation' consideration is not a natural one; there exist instances where the success of a type does not depend at all on how well agents of that type do amongst each other, but only on how well agents of that type do against agents of \emph{other} types. It is dynamics of this type that we consider in the present paper. 

`Purely competitive dynamics', as we shall term them, arise most directly in a knock-out tournament where a large population comprising a certain number of types is paired up in each round, with each pair engaging in a once-off game from which only one of each pair advances, with a contestant's probability of advancing depending on his relative payoff in his engagement. Here, the success of a given type (as measured by the proportion of remaining contestants that are of that type) is invariant to how agents of that type fare against each other. Purely competitive dynamics could also model a learning rule where, in a population of distinct behavioral types, offspring emulate the behavior of their more successful parent. If an offspring's parents are of the same type, it is clear that the offspring will be of that type; there is propagative competition only between parents of differing types. Later in the paper, we shall discuss these examples in the specific context of our model.
  
A relatively new consideration in the game theory literature is that of the effect of mutations (generally speaking, the random changing of types within the population, though mutations have been variously defined) on the dynamics of evolutionary models. The interest in the effect of mutations lies not just in their natural applications -- mutation at a genetic level is one of the key features driving biological evolution \cite{Hartl:1997} -- but also in their interpretation in terms of imperfect memory, experimentation, and bounded rationality \cite{kandori1993}. Mutations have been shown to have qualitatively significant effects on the long-run dynamics of a number of models \cite[see e.g. Refs.][]{rowe1985,bomze1985,fudenberg2006}; in this paper, we shall give some attention to the effect of mutations on our competitive dynamics.

The balance of the paper is structured as follows. Sec.\ \ref{sec:model} sets out the general model -- a sensible characterization of purely competitive dynamics as described above -- while the sections thereafter investigate the behavior of the population state under the dynamics in the context of $3$-type games (governed by $3\times3$ matrices), which have the simultaneous benefits of offering a rich set of possible dynamics and of being easily characterized and illustrated. Sec.\ \ref{sec:zero_mu} analyzes the dynamics without mutations; we show that, under some simplifying assumptions, the behavior of the population state can, qualitatively, be fully characterized.  In Sec.\ \ref{sec:mu}, we investigate the effect of mutations on our dynamics. We show that even low mutation rates can have a nontrivial effect on the behavior of the system; the nature of this effect is typically to stabilize the system in a way that, again, is qualitatively understood. Finally, Sec.\ \ref{sec:conc} concludes, and offers some possibilities for future study.

\section{The model} \label{sec:model}
\subsection{Description of model}\label{sec:modDesc}
We assume a population of $N$ agents, each of which is assigned a pure strategy from the strategy set \hbox{$S = \{s_1, s_2, s_3\}$}; an agent will be said to be `of type $i$' if it is assigned $s_i$. The state of the population in period $t$ is defined by the proportions of the respective types in the population; we shall denote the population state at time $t$ by \hbox{$\mathbf{p}_t := (p_{1t},p_{2t},p_{3t}) \in \Delta^3$}. The $3\times3$ payoff matrix for the game to which the strategies are relevant is given by $\mathbf{U}$, so that the payoff to a type $i$ agent in a once-off play with a type $j$ agent is $u_{ij}$. We require that the elements of $\mathbf{U}$ be strictly positive, and that no two opposite elements be equal (the first of these assumptions is fundamental, as will shortly be seen; the second is not as important, but is rather for convenience).

In each period, a proportion of the population is selected and randomly arranged into pairs. The proportion of the population that is selected in a given period will depend on the individual probabilities of selection of the three types, which in turn depend on the population state in the period. Defining $\sigma_i:\Delta^3 \rightarrow [0,1]$, we denote by $\phi_i(\mathbf{p}_t)$ the probability that a type $i$ agent is selected for pairing in period $t$. Two possibilities immediately present themselves for consideration. The first is $\phi_i \equiv 1 \; \forall\;i$, i.e.\ that in each period, regardless of the population state, each agent is selected for pairing. In a biological context, this would be referred to as a population in which there is zero `selection pressure' \cite{Wright:1937}. The second is that an agent's probability of selection is proportional to its expected payoff -- the strength or fitness of its type in the given population state -- so that $\phi_i(\mathbf{p}) \propto \mathbf{e}_i\mathbf{U}\mathbf{p}$ ($\mathbf{e}_i \in \mathbb{R}^3$ denotes the $i^\textrm{th}$ canonical basis element of $\mathbb{R}^3$).

Once the pairs are selected, each pair engages in a once-off play of the game, the outcome of which is the production of two offspring, who replace their parents in the population and enter into the next period (so that the population size remains unchanged). The two offspring of a given pairing are of the same type, with that type determined as follows:
\begin{equation*}
\text{prob}\,[\text{offsp. type } k \; | \; \text{parents type } i, j] = \frac{\delta_{ik}u_{ij}^\theta + \delta_{jk}u_{ji}^\theta}{u_{ij}^\theta + u_{ji}^\theta},
\end{equation*}
where $\theta>0$ is a constant, and $\delta_{ab}$ is the Kronecker delta ($\delta_{ab}=1$ if $a=b$, else $\delta_{ab}=0$). It should be noted here that this expression is not arbitrary; it is motivated primarily by the intuitive appeal of two cases, $\theta=1$ and \hbox{$\theta=\infty$}. Taking $\theta=1$, as we shall throughout this paper, the probability that the offspring of a pairing of a type $1$ agent and a type $2$ agent are of type 1, for example, is proportional to the payoff the type 1 parent receives in the pairing. On the other hand, in the limiting case $\theta = \infty$, the offspring assume with certainty the type of the parent with the higher payoff. The parameter $\theta$ therefore governs the transmission from relative strength in a mating pair to probability of propagating one's type in mating.

Often, studies of evolutionary dynamics allow for mutation. We generalize our model by assuming that each offspring, after it inherits a type from its parents but before it enters the population for the next period, changes its type with (fixed) probability $\mu\geq0$, with each of the other two types equally likely. 

In the analysis that follows, for convenience, we shall assume an infinite population, i.e.\ $N=\infty$; the finite-population case is less tractable, and we relegate its consideration to a brief discussion in the concluding section. We shall also focus exclusively on the case $\phi_i \equiv 1$; the case of $\phi_i(\mathbf{p}) \propto \mathbf{e}_i\mathbf{U}\mathbf{p}$ will be investigated in a later study by the authors. Finally, as already mentioned, we assume $\theta=1$ throughout.
\subsection{Possible interpretations} \label{sec:interp}
Though not specifically constituted to emulate a biological process, the model developed above does lend itself to a number of biological interpretations. 

Assume, for example, that each agent represents an animal that is genetically programmed to play some pure strategy, i.e.\ to exhibit a specific behavior or trait. The model could then describe the dynamics of intrasexual selection within a population \cite[see e.g. Refs.][]{Darwin1871,thornhill1983,andersson1994}: a single play of the game pits two animals -- and indirectly, the genes that code for their respective behaviors -- against each other; the animals whose strategies receive higher payoffs within their pairings typically go on to mate and to propagate their genetic material to the next generation, while the weaker animals typically do not. For consistency with the assumption of a fixed population size, we could assume full generational replacement and that victorious animals each produce two offspring, as suggested in Sec.\ \ref{sec:modDesc}; alternatively, we could assume that competition eliminates the losing animals from the population, in which case each victor need only produce a single offspring. (To obviate such contrivances, we might simply generalize the model to allow for a variable population size.) The parameter $\theta$ governs the probability that a weaker competitor beats a stronger competitor: for \hbox{$\theta=\infty$}, stronger competitors will always emerge victorious, whereas for $0<\theta<\infty$, there is a nonzero probability of weaker competitors beating stronger ones (on average, though, stronger competitors will still win more often). A number of interpretations for $\mu>0$ are possible, perhaps the simplest being that a nonzero mutation rate corresponds directly to random mutation at a genomic level.

In any event, the competitive nature of the dynamics is clear: the success of a given strategy depends only on how well it fairs against \emph{other} types. This is a sensible characterization of intrasexual selection: if two animals playing the same strategy are competing for a mate, the payoff within the interaction is irrelevant in the sense that, one way or another, the set of genes that codes for their common strategy (behavior, trait, etc.) will be propagated to the next generation. Of course, if the animals survive for more than one generation, and generations are sufficiently close together, a low common payoff in encounters with their same type might render them weaker in subsequent generations; if this is the case, the replicator dynamics (for example) might be more appropriate.

Though the natural description of our model makes use of biological terminology (`offspring', `parents', `mating', etc.), there are many more general competitive interactions which could conceivably be described by the same model. Consider a boxing tournament, for example. Agents playing the three different strategies could correspond, conveniently, to the three commonly defined styles of boxer:  the `swarmer', the `slugger', and the `out-fighter' \cite{donelson2004}. Boxers are paired up at each stage, with one of the boxers in each pairing advancing to the next round, and the other being eliminated; our dynamics features the sensible condition that a boxer's probability of advancing to the next round is dependent purely on that boxer's strength \emph{relative to} his opponent. (To keep the number of boxers constant at $N$, we could specify that, in advancing to the next round, a boxer earns two fights for that round.)

Furthermore, as already noted, the distinction between mindless and minded rules is not clearly set, and our model can be thought of as describing a process of social or observational learning \cite{miller1941,Schlag1994,Laland:2004}, rather than a mindless competitive interaction (somewhat surprisingly, the modeling of such social behavior is of increasing concern in the context of statistical physics \cite{Castellano2009}).

In the social learning paradigm, players can learn only from existing players in the population, i.e.\ by observing the interaction between other players, or by `asking around'. For example, we might suppose that each generation, new agents (perhaps, but not necessarily, offspring of current agents) enter the population, with each new agent making a once-and-for-all choice of strategy based on an observation of a play of the game between an exiting agent and one other agent drawn randomly from the population. Sensibly, the new agent will be more likely to emulate the strategy which received the higher payoff in the observed interaction: for $\theta=\infty$, the new agent will emulate with certainty the strategy observed to be more successful; $\theta<\infty$ could be explained in terms of imperfect observation, while $\mu>0$ could be explained in terms of experimentation. 

Alternatively, instead of assuming that agents are periodically replaced, we could assume that agents are memoryless or that they consider past experience irrelevant to their current situation: for example, an agent might stick to its strategy until it is paired up with an opponent who receives a higher payoff than it does, in which case it will adopt the strategy of its opponent \cite{fudenberg1998}. 

Given these additional interpretations in an observational learning framework, the dynamics might well be referred to as an `emulation dynamics'; nevertheless, the competitive nature of the dynamics remains clear.
\section{The case of no mutations} \label{sec:zero_mu}
As noted, we make the following assumptions throughout the paper: $\phi_i(\mathbf{p}) = 1 \; \forall \; \mathbf{p}  \in \Delta^3, i = 1, 2, 3$; $\theta = 1$; and $N=\infty$. In this section we shall also assume $\mu = 0$ (no mutations). Some noteworthy features of the game under these assumptions are listed below.
\begin{enumerate}
\item Since $\phi_i(\mathbf{p}) \equiv 1$ for $i = 1,2,3$, the diagonal elements of the payoff matrix $\mathbf{U}$ are irrelevant: when a type $i$ parent meets another type $i$ parent, their offspring are guaranteed to be of type $i$. Thus, we are free to normalize the diagonal elements of $\mathbf{U}$ as we wish -- the dynamics will be unaffected.
\item Multiplying opposite elements of $\mathbf{U}$ by the same factor leaves the effects of the dynamics unchanged; the probability that an offspring is of type $i$, if its parents are of types $i$ and $j$, is invariant to multiplying $u_{ij}$ and $u_{ji}$ by some factor $\lambda$. This factor can be unique for different pairs of opposite elements. Thus, we are free to scale opposite elements in the payoff matrix without fear of altering the dynamics.
\end{enumerate}

It will turn out that a convenient normalization of the payoff matrix $\mathbf{U} \rightarrow \tilde{\mathbf{U}}$ is given by:
\begin{equation*}
\tilde{u}_{ij} := \frac{u_{ij}}{u_{ij} + u_{ji}}.
\end{equation*}
This normalization obeys the conditions set out in notes (1) and (2) above -- it is equivalent to setting all diagonal elements to $\tfrac{1}{2}$, and scaling off-diagonal opposite elements $u_{ij}$ and $u_{ji}$ by the factor $\lambda = (u_{ij} + u_{ji})^{-1}$ -- so that the trajectories of the population state under $\mathbf{U}$ will be identical to those under $\tilde{\mathbf{U}}$. The benefit of this normalization is that, for $i\neq j$, $\tilde{u}_{ij}$ is simply the probability that a pairing of a type $i$ agent and a type $j$ agent will yield offspring of type $i$.
\subsection{Stage-to-stage dynamics}
Assume that the population state at time $t$ is $\mathbf{p}_t = (p_{1t},p_{2t},p_{3t})$. Consider the possible pairings that could give rise to type 1 offspring, assuming that there are no mutations ($\mu = 0$):
\begin{itemize}
\item type 1 - type 1, giving rise to type $1$ offspring with probability 1 -- the probability that a random pairing is of this type is $p_{1t}^2$;
\item type 1 - type 2, giving rise to type $1$ offspring with probability $\tilde{u}_{12}$ -- the probability that a random pairing is of this type is $2p_{1t}p_{2t}$; and
\item type 1 - type 3, giving rise to type $1$ offspring with probability $\tilde{u}_{13}$ -- the probability that a random pairing is of this type is $2p_{1t}p_{3t}$.
\end{itemize}

In an infinite population, the proportion of each pairing possibility amongst all other possibilities corresponds exactly with its probability; also, for a given pairing type-combination, the proportion of such meetings that result in type 1 offspring corresponds exactly to the probability that a pairing of that type-combination yields type 1 offspring. Thus, in the case of an infinite population with no mutations, we have the following \emph{exact} update formula for $p_1$:
\begin{equation*}
p_{1,t+1} = p_{1t}^2 + 2p_{1t}p_{2t}\tilde{u}_{12} + 2p_{1t}p_{3t}\tilde{u}_{13}.\\
\end{equation*}
(For the finite population case, we would have to make use of expected value operators.)

The update formulae for $p_2$ and $p_3$ are reached similarly, and we arrive at the update formula for the population state as a whole:\\
\begin{equation}
\left\{
\begin{array}{l}
p_{1,t+1} = p_{1t}^2 + 2p_{1t}p_{2t}\tilde{u}_{12} + 2p_{1t}p_{3t}\tilde{u}_{13}\\
\\
p_{2,t+1} = p_{2t}^2 + 2p_{2t}p_{1t}\tilde{u}_{21} + 2p_{2t}p_{3t}\tilde{u}_{23}\\
\\
p_{3,t+1} = p_{3t}^2 + 2p_{3t}p_{1t}\tilde{u}_{31} + 2p_{3t}p_{2t}\tilde{u}_{32}.\end{array} \right. \label{update1}
\end{equation}\\

In Sec.\ \ref{sec:mu} we shall augment the above update formula to account for mutations; in that case, offspring of a given type can result from \emph{any} pairing. 
\subsection{Fixed points and their stability}\label{sec:fps}
A fixed point is a population state which, under the dynamics, persists through time.
\begin{mydef}
The population state $\mathbf{p}$ is a fixed point if
\begin{equation*}
\mathbf{p}_t = \mathbf{p} \quad \Rightarrow \quad \mathbf{p}_{t+1} = \mathbf{p}.
\end{equation*}
\end{mydef}
Note that the assumption of an infinite population allows us to ignore the stochastic element of the dynamics in characterizing such points.

In discussing the stability of fixed points, a key concept is that of `evolutionary stability' \cite{smith1982,eshel1996}. Loosely, a fixed point exhibits evolutionary stability if, after the introduction into the population of any group of agents of slightly different proportional make-up, the population state moves back to the initial population state (the fixed point). In other words, for all perturbations smaller than a given size, the population state eventually returns to the fixed point.

\begin{mydef}
Let $\mathbf{p}$ be a fixed point. $\mathbf{p}$ is \textit{evolutionarily stable} (hereafter `ES') if $\exists \; \varepsilon >0$ such that
\begin{equation*}
\mathbf{p}_t \in B_{\Delta^3}(\mathbf{p},\varepsilon) \quad \Rightarrow \quad \lim_{\tau \rightarrow \infty}\mathbf{p}_{t+\tau} = \mathbf{p}.
\end{equation*}
\end{mydef}
(Here, as hereafter, we denote by $B_{\Delta^n}(\mathbf{p},\varepsilon)$ the open ball on $\Delta^n$ of radius $\varepsilon$, centered at $\mathbf{p} \in \Delta^n$; that is, $B_{\Delta^n}(\mathbf{p},\varepsilon) := \left\{\mathbf{x}\in \Delta^n: ||\mathbf{x} - \mathbf{p}|| < \varepsilon \right\}$, where $||\cdot||$ is the Euclidean norm in $\mathbb{R}^n$.) 

Finally, we shall call a fixed point `globally attracting' if, for any initial population state away from the other fixed points, the population state converges to that fixed point.
\begin{mydef}
Let $\mathbf{p}$ be a fixed point, along with $\mathbf{p^1},\mathbf{p^2}, \ldots, \mathbf{p^m}$. $\mathbf{p}$ is \textit{globally attracting} if:
\begin{equation*}
\mathbf{p}_t \in \Delta^3\backslash \{\mathbf{p^1}, \ldots, \mathbf{p^m}\}  \quad \Rightarrow \quad  \lim_{\tau\rightarrow\infty}\mathbf{p}_{t+\tau} = \mathbf{p}.
\end{equation*}
\end{mydef}
It is clear that any globally attracting isolated fixed point is also ES.

We shall begin our identification and classification of the fixed points with the following trivial result:
\begin{thm}\label{pureStatesFixed_thm}
If $\mu = 0$, the population states $\mathbf{e_1}, \mathbf{e}_2$ and $\mathbf{e}_3$ are fixed points.
\end{thm}
\begin{proof}
If there are only type $i$ agents in the population ($\mathbf{p} = \mathbf{e}_i$), pairings comprise only type $i$ agents, and thus, with no mutations, only type $i$ offspring will be produced. This state will persist.
\end{proof}

We shall now consider the evolutionary stability of these `pure' states. We shall describe the conditions under which a pure state can be ES, and show that ES pure states are in fact globally attracting as well. We shall then show that at most one of the pure states can be ES, and describe the conditions under which none of the pure states is ES.
\begin{thm}\label{dominantES_thm}
The fixed point $\mathbf{p} = \mathbf{e}_i$ is ES if, and only if, $\tilde{u}_{ij} > \tfrac{1}{2} \; \; \forall \; \; j\neq i$ (equivalently, $u_{ij} > u_{ji} \; \; \forall \; \; j\neq i$). Moreover, if $\tilde{u}_{ij} > \tfrac{1}{2} \; \; \forall \; \; j\neq i$, then $\mathbf{e}_1, \mathbf{e}_2, \mathbf{e}_3$ are the only fixed points, and $\mathbf{p} = \mathbf{e}_i$ is globally attracting.
\end{thm}
\begin{proof}
($\Leftarrow$) Without loss of generality, assume $\tilde{u}_{12}, \tilde{u}_{13} > \tfrac{1}{2}$, and let the current population state be $\mathbf{p}_t = (p_{1t},p_{2t},p_{3t}) \neq \mathbf{e}_2,\mathbf{e}_3$. We shall demonstrate that $\displaystyle \lim_{\tau \rightarrow \infty} p_{1,t+\tau} = 1$; that is, the population returns to the pure state $\mathbf{e}_1$.

From \eqref{update1}, we have the following update formula for $p_1$:
\begin{align*}
p_{1,t+1} &= p_{1t}^2 + 2p_{1t}p_{2t}\tilde{u}_{12} + 2p_{1t}p_{3t}\tilde{u}_{13}\\
&= p_{1t}\left(p_{1t} + 2p_{2t}\tilde{u}_{12} + 2p_{3t}\tilde{u}_{13}\right)\\
&\geq p_{1t}\left(p_{1t} + 2p_{2t}(\tfrac{1}{2}) + 2p_{3t}(\tfrac{1}{2})\right)\\
&= p_{1t}\left(p_{1t} + p_{2t} + p_{3t}\right) = p_{1t}.
\end{align*}
We have equality in the third line if, and only if, $p_{2t} = p_{3t} = 0$ (we have already assumed $p_{1t} \neq 0$). So for all non-pure population states at time $t$, the $p_1$ component is strictly monotonically increasing over time, so none is a fixed point, and thus $\displaystyle \lim_{\tau \rightarrow \infty} p_{1,t+\tau} = 1$. The population state $\mathbf{e}_1$ is therefore globally attracting, and by implication, ES too.\\

($\Rightarrow$) Suppose, without loss of generality, that $\tilde{u}_{12} < \tfrac{1}{2}$. We shall show that the pure state $\mathbf{p} = \mathbf{e}_1$ is not resistant to the invasion of type 2 agents; in particular, we shall show that, for any population state $\mathbf{p}_t = \xi\mathbf{e}_2 + (1-\xi)\mathbf{e}_1$ where $\xi\in (0,1)$, the component $p_1$ decreases in time.

The update formula for the $p_1$ component of $\mathbf{p}$ in a population where $p_{1t},p_{2t} > 0$ and $p_{3t} = 0$ is given by:
\begin{align*}
p_{1,t+1} &= p_{1t}^2 + 2p_{1t}p_{2t}\tilde{u}_{12}\\
&= p_{1t}\left(p_{1t} + 2p_{2t}\tilde{u}_{12}\right)\\
&< p_{1t}\left(p_{1t} + 2p_{2t}(\tfrac{1}{2})\right)\\
&= p_{1t}\left(p_{1t} + p_{2t}\right) = p_{1t};
\end{align*}
that is, $p_1$ is strictly monotonically decreasing in time. So, any perturbation away from the initial state $\mathbf{p} = \mathbf{e}_1$ along that support of $\Delta^3$ where $p_3 = 0$ will ensure that the population state does not return to $\mathbf{e}_1$. Since any open ball $B_{\Delta^3}(\mathbf{e}_1,\varepsilon)$ contains such points, the pure state $\mathbf{p} = \mathbf{e}_1$ is not ES.
\end{proof}

Fig.\ \ref{fig1}(a) on pg.\ \pageref{fig1} contains a phase portrait showing some typical trajectories for the case of an ES pure state. In this `ternary diagram', as hereafter, the right, top and left vertices of the simplex correspond to the pure population states $\mathbf{e}_1, \mathbf{e}_2$ and $\mathbf{e}_3$ respectively. The scaling for each component of the population state is along the perpendicular dropped from the vertex corresponding to its pure state to the side opposite that vertex.

\begin{figure*}[t]
\centering
\includegraphics[width=155mm]{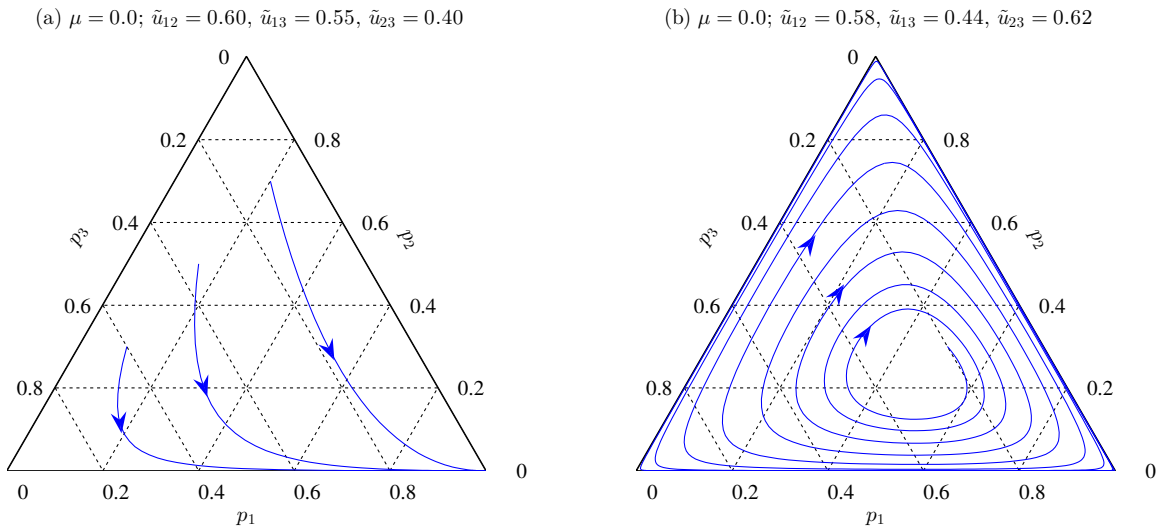}
\caption{A few representative trajectories for the case (a) where the existence of a dominant type leads to an ES pure state; and a typical trajectory for the case (b) where no dominant type exists, resulting in a locally-unstable interior fixed point.\label{fig1}}
\end{figure*}

The case $\tilde{u}_{12}, \tilde{u}_{13} > \tfrac{1}{2}$, which the previous theorem shows to be both necessary and sufficient for the pure state $\mathbf{p} = \mathbf{e}_1$ to be ES, does not permit any other pure state to be ES. (For $\mathbf{e}_2$ to be ES, the above theorem would require $\tilde{u}_{21} > \tfrac{1}{2}$, which contradicts $\tilde{u}_{12} > \tfrac{1}{2}$. The requirements for $\mathbf{e}_3$ to be ES lead to a similar contradiction.) So we can have at most one pure state ES fixed point. We shall call the type associated with such a state the `dominant type'; the conditions required of the payoff matrix for its existence and the result in Thm.\ \ref{dominantES_thm} justify this nomenclature.

It is interesting to note that the evolutionarily stable fixed point in this case need not correspond to a Nash equilibrium of the static game; indeed, the globally attracting pure state can even correspond to a strategy that is strictly dominated in the static game! Take, for example, the following payoff matrix:
\begin{equation*}
{\mathbf{U}} = \left( {\begin{array}{*{20}{c}}
   2 & 6 & 7  \\
   4 & 8 & 8  \\
   3 & 2 & 5  \\
 \end{array} } \right) .
\end{equation*}
Under our dynamics, the pure state of type 1 agents is globally attracting, despite the fact that strategy 1 is strictly dominated by strategy 2. This may be understood by noting that the competitive dynamics requires us to consider opposite elements in the payoff matrix rather than those above and below each other, as the traditional concept of dominance does.

Another consequence of Thm.\ \ref{dominantES_thm} is that there exist cases where no ES pure state exists. These are the cases where none of the types is dominant:
\begin{itemize}
\item[(i)]$\tilde{u}_{12} > \tfrac{1}{2}, \tilde{u}_{13} < \tfrac{1}{2}, \tilde{u}_{23} > \tfrac{1}{2};$
\item[(ii)]$\tilde{u}_{12} < \tfrac{1}{2}, \tilde{u}_{13} > \tfrac{1}{2}, \tilde{u}_{23} < \tfrac{1}{2}.$
\end{itemize}

A suitable permutation of the indices equates the two cases, so we can take, say, the first to be general. Here, type 1 `beats' type 2, which `beats' type 3, which `beats' type 1. In this case, the game mirrors the classic rock-paper-scissors game of the game theory literature -- see, for example, \citet{hofbauer1998}.

We now turn our attention to non-pure state fixed points. We know from Thm.\ \ref{dominantES_thm} that such fixed points are not possible if one of the fixed points is ES, so we need only consider the case where no ES pure state exists, the general form of which we have taken to be case (i) above.

The theorem below establishes that, in such a case, a unique non-pure state fixed point always exists, and that it lies in the interior of $\Delta^3$. The subsequent theorem will show that this interior fixed point can not be ES.
\begin{thm}
\label{mixedFixed_thm}
If $\tilde{u}_{12} > \tfrac{1}{2}, \tilde{u}_{13} < \tfrac{1}{2}, \tilde{u}_{23} > \tfrac{1}{2}$, then:
\begin{itemize}
\item[(a)] $\exists! \; \mathbf{p} \in \Delta^3\backslash\{\mathbf{e}_1,\mathbf{e}_2,\mathbf{e}_3\}$ such that $\mathbf{p}_t = \mathbf{p} \; \Rightarrow \; \mathbf{p}_{t+1} = \mathbf{p}$;
\item[(b)] $\mathbf{p} \in \text{\emph{int}}(\Delta^3)$.
\end{itemize}
\end{thm}

\begin{proof}
We shall derive the conditions for such a fixed point to exist, and demonstrate that a unique point in $\Delta^3\backslash\{\mathbf{e}_1,\mathbf{e}_2,\mathbf{e}_3\}$ fulfils these conditions. Assume that $\mathbf{p} \in \text{int}(\Delta^3)$ is our candidate for an interior fixed point. (We shall consider the case where $\mathbf{p} \in \text{bd}(\Delta^3)\backslash\{\mathbf{e}_1,\mathbf{e}_2,\mathbf{e}_3\}$ later, and demonstrate that such a point can not be fixed.) Applying the update formula for $\mathbf{p}_{t+1}$, and substituting in the fixed point condition that $\mathbf{p}_{t+1} = \mathbf{p}_t = \mathbf{p}$, we arrive at the following system of equations:
\begin{align*}
&\left\{
\begin{array}{l}
p_{1} = p_{1}^2 + 2p_{1}p_{2}\tilde{u}_{12} + 2p_{1}p_{3}\tilde{u}_{13}\\
p_{2} = p_{2}^2 + 2p_{2}p_{1}\tilde{u}_{21} + 2p_{2}p_{3}\tilde{u}_{23}\\
p_{3} = p_{3}^2 + 2p_{3}p_{1}\tilde{u}_{31} + 2p_{3}p_{2}\tilde{u}_{32}\\
\end{array} \right.\\
\Rightarrow \; \; &\left\{
\begin{array}{l}
1 = p_{1} + 2p_{2}\tilde{u}_{12} + 2p_{3}\tilde{u}_{13}\\
1 = p_{2} + 2p_{1}\tilde{u}_{21} + 2p_{3}\tilde{u}_{23}\\
1 = p_{3} + 2p_{1}\tilde{u}_{31} + 2p_{2}\tilde{u}_{32}\\
\end{array} \right.\\
\Rightarrow \; \; &2\tilde{\mathbf{U}}\mathbf{p} = (1,1,1),
\end{align*}
where we used the fact that the components of $\mathbf{p}$ are nonzero in dividing them out in the second step. Taking into account the conditions $p_{1}+p_{2}+p_{3} = 1$ and $\tilde{u}_{ij} + \tilde{u}_{ji} = 1$, we arrive at the following linear system:
\begin{equation*}
\left(
\begin{array}{cccc}
1/2 							& \tilde{u}_{12} 		& \tilde{u}_{13} \\
1-\tilde{u}_{12} 	& 1/2 							& \tilde{u}_{23} \\
1-\tilde{u}_{13} 	& 1-\tilde{u}_{23} 	& 1/2				\\
1									&	1									&	1	
\end{array} \right)\left(
\begin{array}{c}
p_{1}\\
p_{2}\\
p_{3}
\end{array} \right) = \left(
\begin{array}{c}
1/2\\
1/2\\
1/2\\
1
\end{array} \right).
\end{equation*}\\
The above system turns out to be consistent; the unique solution is given by:
\begin{align}
p_{1} = \frac{1-2\tilde{u}_{23}}{1-2(\tilde{u}_{12} - \tilde{u}_{13} + \tilde{u}_{23})}, \nonumber\\
p_{2} = \frac{2\tilde{u}_{13} - 1}{1-2(\tilde{u}_{12} - \tilde{u}_{13} + \tilde{u}_{23})} \label{intfixed1},\\
p_{3} = \frac{1-2\tilde{u}_{12}}{1-2(\tilde{u}_{12} - \tilde{u}_{13} + \tilde{u}_{23})}. \nonumber
\end{align}
It remains to be shown that $p_{1},p_{2},p_{3}>0$. Note that $\tilde{u}_{12} - \tilde{u}_{13} + \tilde{u}_{23} > \tfrac{1}{2}$, so that $1-2(\tilde{u}_{12} - \tilde{u}_{13} + \tilde{u}_{23}) < 0$. We thus require the numerators above to be negative; it is readily verified that this is the case. Thus, $\mathbf{p}\in \text{int}(\Delta^3)$, and so, under the assumptions concerning $\tilde{\mathbf{U}}$ that forbid an ES pure state, there always exists a unique fixed point in the interior of $\Delta^3$; it is given by \eqref{intfixed1}.

Finally, we consider the case where $\mathbf{p} \in \text{bd}(\Delta^3)\backslash\{\mathbf{e}_1,\mathbf{e}_2,\mathbf{e}_3\}$. Assume, without loss of generality, that $p_3 = 0$, so that $p_1, p_2 > 0$, with $p_1 + p_2 = 1$. Our system of equations incorporating the fixed point condition is then 
\begin{equation*}
\left\{
\begin{array}{l}
p_{1} = p_{1}^2 + 2p_{1}p_{2}\tilde{u}_{12}\\
p_{2} = p_{2}^2 + 2p_{2}p_{1}\tilde{u}_{21}\\
1 = p_1 + p_2\\
\end{array} \right.\\
\Rightarrow \; \; \left\{
\begin{array}{l}
(1-2\tilde{u}_{12})p_1 = 1-2\tilde{u}_{12}\\
(1-2\tilde{u}_{12})p_1 = 0.\\
\end{array} \right.\\
\end{equation*}
This is consistent only if $\tilde{u}_{12} = \tfrac{1}{2}$, which we have assumed not to be the case. Thus, the only fixed points on the boundary of $\Delta^3$ are the pure states.
\end{proof}

So we have found the conditions under which an interior fixed point exists; it remains to characterize the stability of this interior fixed point. The following result establishes that the interior fixed point is never ES.
\begin{thm}\label{mixedNotES_thm}
If an interior fixed point $\mathbf{p}$ exists, then $\forall \; \varepsilon>0, \exists \; {\mathbf{p}^\prime} \in B_{\Delta^3}(\mathbf{p},\varepsilon)$ such that $\displaystyle \mathbf{p}_t = {\mathbf{p}^\prime} \; \Rightarrow \; \lim_{\tau\rightarrow \infty}\mathbf{p}_{t+\tau} \neq \mathbf{p}$.
\end{thm}

\begin{proof}
Suppose, without loss of generality, that $\tilde{u}_{12} > \tfrac{1}{2}, \tilde{u}_{13} < \tfrac{1}{2}, \tilde{u}_{23} > \tfrac{1}{2}$, so that we have the interior fixed point given by \eqref{intfixed1}, which we call $\mathbf{p}\in \text{int}(\Delta^3)$. 

Taking into account the restrictions that the proportions sum to unity, and that opposite elements of $\tilde{\mathbf{U}}$ sum to unity, the third equation in \eqref{update1} becomes redundant, and the system may be written in two dimensions:
\begin{align}
p_{1,t+1}&= p_{1t}\left[p_{1t}(1-2\tilde{u}_{13}) + 2p_{2t}(\tilde{u}_{12} -\tilde{u}_{13}) + 2\tilde{u}_{13}\right] \nonumber \\
&=: F_1(p_{1t},p_{2t});\nonumber\\
p_{2,t+1}&= p_{2t}\left[2p_{1t}(1-\tilde{u}_{12} - \tilde{u}_{23}) + p_{2t}(1-2\tilde{u}_{23}) + 2\tilde{u}_{23}\right] \nonumber\\
&=: F_2(p_{1t},p_{2t}). \nonumber
\end{align}
$\mathbf{p}$, understood in this context to be the vector comprising the first two elements of the original three-element $\mathbf{p}$, is a fixed point of this system. To determine the stability of this fixed point, we employ the usual linearization method \cite[e.g. Ref.][]{Kocic:1993}: if we can show that at least one of the eigenvalues of the Jacobian matrix $\partial (F_1,F_2)/\partial (p_{1t},p_{2t})$, evaluated at the fixed point, has modulus greater than unity, then the fixed point is unstable. Defining  $D:= 1-2(\tilde{u}_{12}-\tilde{u}_{13}+\tilde{u}_{23}) < 0$, these eigenvalues may, with some effort, be shown to be:
\begin{equation*}
\lambda_{1,2} = 1 \pm i\sqrt{\frac{(2\tilde{u}_{12} - 1)(1-2\tilde{u}_{13})(2\tilde{u}_{23} - 1)}{-D}} = 1 \pm i\sqrt{\omega},
\end{equation*}
where $0<\omega := -(2\tilde{u}_{12} - 1)(1-2\tilde{u}_{13})(2\tilde{u}_{23} - 1)/D<1$. So, the eigenvalues are distinct and complex, with modulus greater than unity:
\begin{equation*}
|\lambda_{1,2}|^2 = 1 + \omega > 1.
\end{equation*}
The fixed point is therefore locally unstable, i.e.\ not ES, and the trajectories nearby are outward spirals. On the unit simplex $\Delta^3$, these trajectories translate into outward spirals.
\end{proof}

A typical outward-spiraling trajectory is presented in Fig.\ \ref{fig1}(b).

A numerical investigation such as that carried out in the following section (and explained in greater detail there) reveals this outward spiral to have as a limit cycle the boundary of the simplex. The limiting behavior of the system, then, is oscillatory: the proportion of type 1 agents increases to nearly unity (with a near zero proportion of type 2 and type 3 agents), after which the proportion of type 2 agents increases to nearly unity, after which the proportion of type 3 agents increases to nearly unity, and so on. This oscillatory long-run behavior of the population state is similar to that under the replicator dynamics for rock-paper-scissors games \cite{hofbauer1998}; this is unsurprising because, as previously noted, the payoff logic behind the current case in our dynamics mirrors that of the rock-paper-scissors game.

Under some observational learning interpretations of our model, this oscillatory behavior might be regarded as unrealistic: rational agents should eventually realize that the system is locked in a perpetual cycle, so that more sophisticated learning rules would be needed for the system to converge \cite{fudenberg1998}.

\section{The case of a nonzero mutation rate} \label{sec:mu}
We now consider the effect of a nonzero mutation rate $\mu > 0$, maintaining our assumptions of an infinite population and non-discriminatory selection. The salient difference between this case and the previous case, where we assumed no mutations, is that type $i$ offspring can result from \emph{any} pairing. There are, for example, two ways in which a type 1 offspring can result from the pairing of a type 1 agent with a type 2 agent: the parents can initially produce type 1 offspring (probability $\tilde{u}_{12}$) which do not subsequently mutate (probability $1-\mu$), or they can initially produce type 2 offspring (probability $\tilde{u}_{21}$) which mutate into type 1 offspring (probability $\mu/2$, since it is equally likely that they will mutate into type 3 agents). Considering all possible pairings by which offspring of each type may be produced, and again relying on the infinite population assumption, we obtain the following (simplified) update formula for the population state:
\begin{equation}
\left\{
\begin{array}{l}
p_{1,t+1} = \left(1-\frac{3\mu}{2}\right)\left(p_{1t}^2 + 2p_{1t}p_{2t}\tilde{u}_{12} + 2p_{1t}p_{3t}\tilde{u}_{13}\right) + \frac{\mu}{2}\\
\\
p_{2,t+1} = \left(1-\frac{3\mu}{2}\right)\left(p_{2t}^2 + 2p_{1t}p_{2t}\tilde{u}_{21} + 2p_{2t}p_{3t}\tilde{u}_{23}\right) + \frac{\mu}{2}\\
\\
p_{3,t+1} = \left(1-\frac{3\mu}{2}\right)\left(p_{3t}^2 + 2p_{1t}p_{3t}\tilde{u}_{31} + 2p_{2t}p_{3t}\tilde{u}_{32}\right) + \frac{\mu}{2}. \end{array} \right. \label{update_mut}
\end{equation}\\

In a biological context, we can expect mutation rates to be very low, i.e.\ $\mu\ll1$ \cite[see e.g.][]{Nachman:2000,Haupt:2004}; here, the weakest restriction we shall assume is $\mu < 2/3$, though typical mutation rates will be much lower. Note that if $\mu>2/3$, we have a perverse scenario where complete dominance of one type in a given period (say, $p_{1t} = 1$) results in that type being the minority in the next period ($p_{1,t+1} < p_{2,t+1},p_{3,t+1}$). We shall see that the nonzero mutation rate renders finding and characterizing the fixed points more complicated than before -- for example, the $\mu/2$ term at the far right of each equation in \eqref{update_mut} prevents us from factorizing $p_{it}$, if nonzero, out of the right hand side of the $i^\textrm{th}$ equation.

Before a more rigorous analysis, we might make some initial remarks and hypotheses regarding the fixed points of the system, and their stability, under \eqref{update_mut}:
\begin{enumerate}
\item The pure states $\mathbf{p} = \mathbf{e}_i$ are no longer fixed points: a population comprising only type $i$ agents, will, in the initial reproduction stage, produce only type $i$ agents, but will see a proportion $\mu$ of these offspring mutate to other types, rendering the resultant population state `mixed'.
\item If $\tilde{U}$ is such that $\mathbf{e}_i$ would be an ES fixed point in the case without mutations (Thm.\ \ref{dominantES_thm}), we might expect there to be an ES fixed point (or at least an invariant set) near $\mathbf{e}_i$ under the dynamics that allows for mutations.
\item If $\tilde{U}$ is such that there would be an interior fixed point in the case without mutations (Thm.\ \ref{mixedFixed_thm}), then we might expect a similar fixed point in the case with mutations. It is not clear whether this point, if it were to exist, would be stable or not.
\end{enumerate}

We begin the more detailed analysis by rewriting \eqref{update_mut}, taking into account the restrictions $p_1 + p_2 + p_3 = 1$ and $\tilde{u}_{ij} = 1-\tilde{u}_{ji}$; under these restrictions, the third equation turns out to be solved identically if the first two are solved, and the system simplifies to:
\begin{equation}
\left\{
\begin{array}{l}
p_{1,t+1} - p_{1t} = p_{1t}\left(1-\frac{3\mu}{2}\right)\Big((1-p_{1t})(2\tilde{u}_{13} - 1)\Big.\\
\textcolor{white}{p_{1,t+1} - p_{1t} =}\frac{\mu}{2}\left(1-3p_{1t}\right)+\Big.2p_{2t}(\tilde{u}_{12} - \tilde{u}_{13})\Big)\\
\\
p_{2,t+1} - p_{2t} = p_{2t}\left(1-\frac{3\mu}{2}\right)\Big(2p_{1t}(1-\tilde{u}_{12} - \tilde{u}_{23})\Big.\\
\textcolor{white}{p_{2,t+1} - p_{2t} =}\frac{\mu}{2}\left(1-3p_{2t}\right) + \Big.(1-p_{2t})(2\tilde{u}_{23}-1)\Big).\\
\end{array} \right.  \label{general_mut}
\end{equation}

\subsection{Existence of feasible fixed points}

For a fixed point, the left hand sides of the above equations are zero; this leaves a system of two non-homogeneous quadratic forms in $p_{1}$ and $p_{2}$ (we may drop the time indices):

\begin{equation}
\left\{
\begin{array}{l}
\left[- \left(1-\frac{3\mu}{2}\right)\left(2\tilde{u}_{13} - 1\right) -2\left(1-\frac{3\mu}{2}\right)(\tilde{u}_{12}-\tilde{u}_{13})p_2\right.\\
\;\;+\frac{3\mu}{2} \left.\right]p_1+\left(1-\frac{3\mu}{2}\right)\left(2\tilde{u}_{13} - 1\right)p_1^2 - \frac{\mu}{2} = 0\\
\\
\left[- \left(1-\frac{3\mu}{2}\right)\left(2\tilde{u}_{23} - 1\right) -2\left(1-\frac{3\mu}{2}\right)(1-\tilde{u}_{12} - \tilde{u}_{23})p_1\right.\\
\;\;+\frac{3\mu}{2} \left.\right]p_2+\left(1-\frac{3\mu}{2}\right)\left(2\tilde{u}_{23} - 1\right)p_2^2 - \frac{\mu}{2} = 0.\\
\end{array} \right.   \label{quad_mut}
\end{equation}

From this system, we expect four solutions for $\mathbf{p} = (p_1,p_2,1-p_1-p_2)$, though of course these need not lie on the unit simplex; this would require $p_1 \geq 0, p_2 \geq 0, p_1+p_2\leq1$.

Though analytical results have been derived, these results are generally intractable (involving enormous expansions which do not seem amenable to simplification even with the aid of a computer algebra system) and thus we shall present numerical results only.

For each setup ($\tilde{u}_{12},\tilde{u}_{13},\tilde{u}_{23},\mu$), we ask two questions:
\begin{enumerate}
\item How many, if any, fixed points lie on the unit simplex $\Delta^3$?
\item For each of the fixed points that do lie on $\Delta^3$, is the fixed point stable or unstable in terms of the definitions set out earlier?
\end{enumerate}

To solve for the fixed points of the system, we note that the above system can be written in the general form
\begin{equation}
\left\{
\begin{array}{l}
a_1 + a_2p_1 + a_3p_1p_2 + a_4p_1^2 = 0\\
b_1 + b_2p_2 + b_3p_1p_2 + b_4p_2^2 = 0.\\
\end{array} \right.   \label{quadgen}
\end{equation}

Assume for now that none of the $p_i,a_i,b_i$ is zero. Then we may solve for $p_2$ in the first equation in \eqref{quadgen}:
\begin{equation}
p_2 = -\frac{1}{a_3}\left(\frac{a_1}{p_1} + a_2 + a_4p_1\right). \label{p_2sub}
\end{equation}
Substituting this result into the second equation in \eqref{quadgen} and simplifying yields a quartic polynomial in $p_1$:
\begin{align}
&a_1^2b_4 + (2a_1a_2b_4 - a_1a_3b_2)p_1 + (a_3^2b_1 - a_2a_3b_2 - a_1a_3b_3\nonumber\\
 &+ a_2^2b_4 + 2a_1a_4b_4)p_1^2 + (2a_2a_4b_4 - a_3a_4b_2 - a_2a_3b_3)p_1^3 \nonumber\\
&+ (a_4^2b_4 - a_3a_4b_3)p_1^4 = 0. \label{quartic}
\end{align}
For each set of coefficients, \eqref{quartic} has four (not necessarily distinct, and not necessarily real) solutions, each of which corresponds to a solution for $p_2$, as determined by \eqref{p_2sub}. $p_3$ is then determined by $p_3 = 1 - p_1 - p_2$.

In this way, using a numerical approach (e.g.\ finding the eigenvalues of the polynomial's characteristic matrix) we may efficiently find, for each combination of parameter values, the roots of \eqref{quad_mut}. We may then investigate whether any of these roots lies on the unit simplex.

We looped the elements of $\tilde{\mathbf{U}}$ using a step size of $0.001$, making sure that none of the assumptions was violated, and we looped $\mu$ between -- but not including -- its assumed limits $(0,2/3)$ using a step size of $0.001$. The key result of this investigation is that, for each setup ( exactly one solution (i.e.\ one fixed point of the system) lies on $\Delta^3$. Moreover, in each case, that solution lies in the \emph{interior} of $\Delta^3$.
\subsection{Stability}
For each setup, we have a unique fixed point on the unit simplex. We may analyze the stability of the fixed point making use of the usual Jacobian analysis, as we did in Sec.\ \ref{sec:fps} - here, though, our analysis will be numerical. For each setup, we shall numerically calculate the eigenvalues of the Jacobian matrix, evaluated at the feasible fixed point. 
\subsubsection{The case of a dominant type}

\begin{figure*}[t]
\centering
\includegraphics[width=155mm]{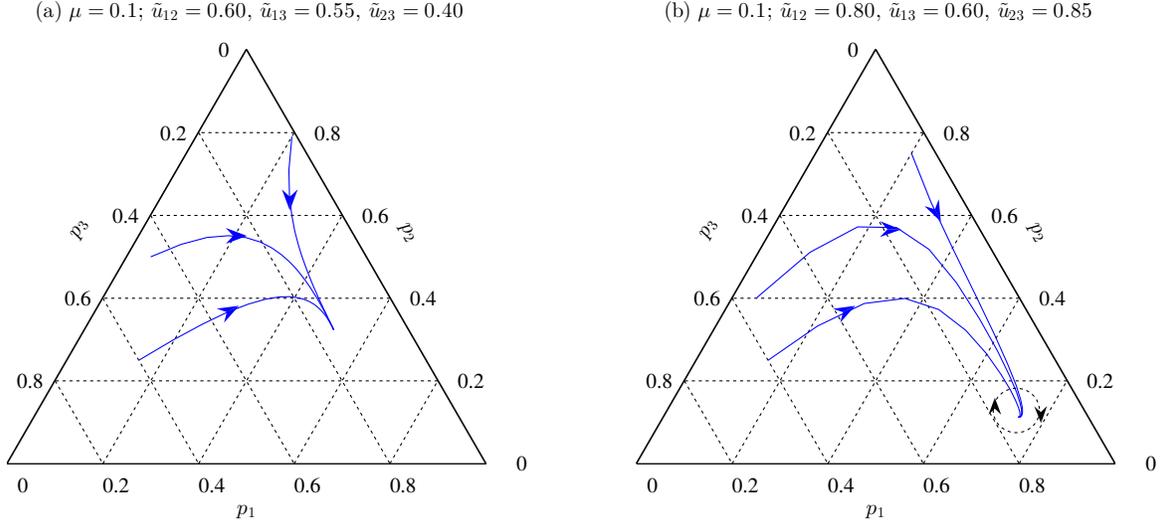}		
\caption{The two types of behavior which can arise in the case of a dominant type. Panel (a) shows typical trajectories for the case where the eigenvalues of the Jacobian, evaluated at the fixed point, are real, with magnitude less than unity; panel (b) shows typical trajectories for the case where the eigenvalues of the Jacobian, evaluated at the fixed point, are complex, with magnitude less than unity. For the latter trajectories, the inward spiral near the fixed point is so small as to be invisible on the plot.\label{fig2}}
\end{figure*}

We took $\tilde{u}_{12},\tilde{u}_{13}>\tfrac{1}{2}$ as the general case that resulted in ES pure states when $\mu = 0$. With no mutations, such setups lead to the population state $\mathbf{p} = \mathbf{e}_1$ being a globally attracting fixed point.

In this and other cases, the limiting case $\mu\rightarrow 0^{+}$ will yield behavior identical to that of the case $\mu=0$. The practical implication of this is that, for sufficiently low mutation rates, the behavior of the system will be practically identical to that with no mutations. 

For sufficiently large mutation rates, however, and for certain setups, the eigenvalues of the Jacobian matrix, evaluated at the interior fixed point, are \emph{complex}, with modulus strictly less than unity, so that the behavior near the fixed point is no longer that of a regularly stable fixed point, but rather an inward spiral. For other setups, though, the eigenvalues are real, both with magnitude lower than one, so that the fixed point is regularly stable, as in the case with no mutations. 

In the case of a dominant type, these are the only behaviors possible under the dynamics; Fig.\ \ref{fig2} illustrates typical trajectories of each.

In Fig.\ \ref{fig3}, for various mutation rates, we plot the regions in the $\tilde{u}_{12}$-$\tilde{u}_{23}$ plane that are associated with fixed points of the two stability types mentioned above. We fix $\tilde{u}_{13} = 0.53$; thus, we are interested in setups with $\tilde{u}_{12}>\tfrac{1}{2}$. Notice that the region for which the `new' behavior occurs, i.e.\ for which the fixed point is no longer regularly stable but rather a stable spiral, grows as we increase $\mu$.

\begin{figure*}[ht!]
\centering
\includegraphics[width=155mm]{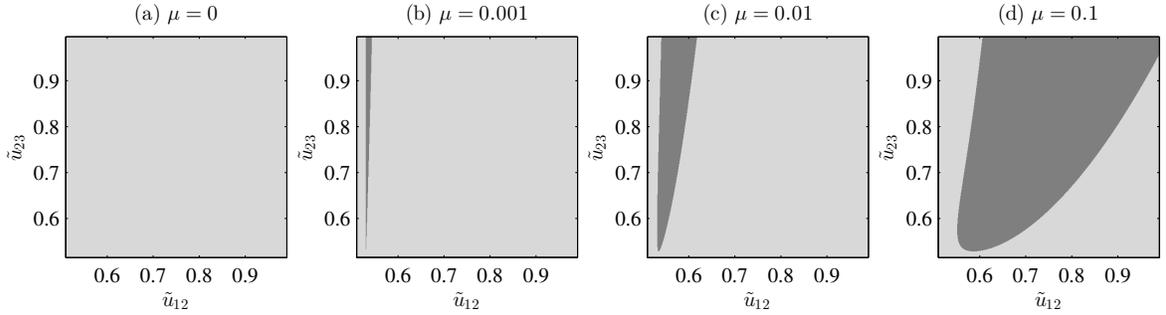}	
\caption{The case of a dominant type. Regions in the $\tilde{u}_{12}$-$\tilde{u}_{23}$ plane which result in fixed points of the various types, for various mutation rates, fixing $\tilde{u}_{13} = 0.53$. The light regions correspond to regularly stable fixed points; the dark regions correspond to stable spirals.\label{fig3}}
\end{figure*}
\subsubsection{The case of no dominant type}
The general case that resulted in an unstable interior fixed point when $\mu = 0$ was taken to be $\tilde{u}_{12}>\tfrac{1}{2},\tilde{u}_{13}<\tfrac{1}{2},\tilde{u}_{23}>\tfrac{1}{2}$. With $\mu = 0$, such setups lead to an unstable interior fixed point; the trajectories are outward spirals, with limit cycles at the boundary of the simplex. For sufficiently low mutation rates, then, this is the behavior of the system as well. (More precisely, for a given payoff matrix $\tilde{\mathbf{U}}$ obeying the above condition, there exists a positive number $\varepsilon$ such that, if $\mu<\varepsilon$, the trajectories around the fixed point are outward spirals whose limit cycle is arbitrarily close to the boundaries of the simplex.) For larger values of $\mu$, the behavior of the system becomes more complex. Three types of behavior are possible; they are described below, along with the nature of the eigenvalues of the Jacobian (evaluated at the fixed point) that diagnoses the behavior near the fixed point in each case. (The behavior away from the fixed point can not be diagnosed using the eigenvalues; for this we rely on numerical results.)
\begin{enumerate}
\item An unstable spiral near the fixed point, spiraling outward to a limit cycle contained in the interior of $\Delta^3$. The diagnostic eigenvalue condition for the behavior near the fixed point is: $\operatorname{Im}{\lambda_{1,2}} \neq 0, |\lambda_{1,2}|>1$.
\item A stable spiral near the fixed point. In fact, trajectories spiral inward to the fixed point from all points in $\Delta^3$, so that it is globally attracting as well. Diagnostic condition: $\operatorname{Im}{\lambda_{1,2}} \neq 0, |\lambda_{1,2}|<1$.
\item Non-spiral inward trajectories towards the fixed point, which we term a `regularly stable fixed point'. Diagnostic condition: $\operatorname{Im}{\lambda_{1,2}} = 0, |\lambda_{1,2}|<1$.
\end{enumerate}

\begin{figure*}[t]
\centering	
\includegraphics[width=155mm]{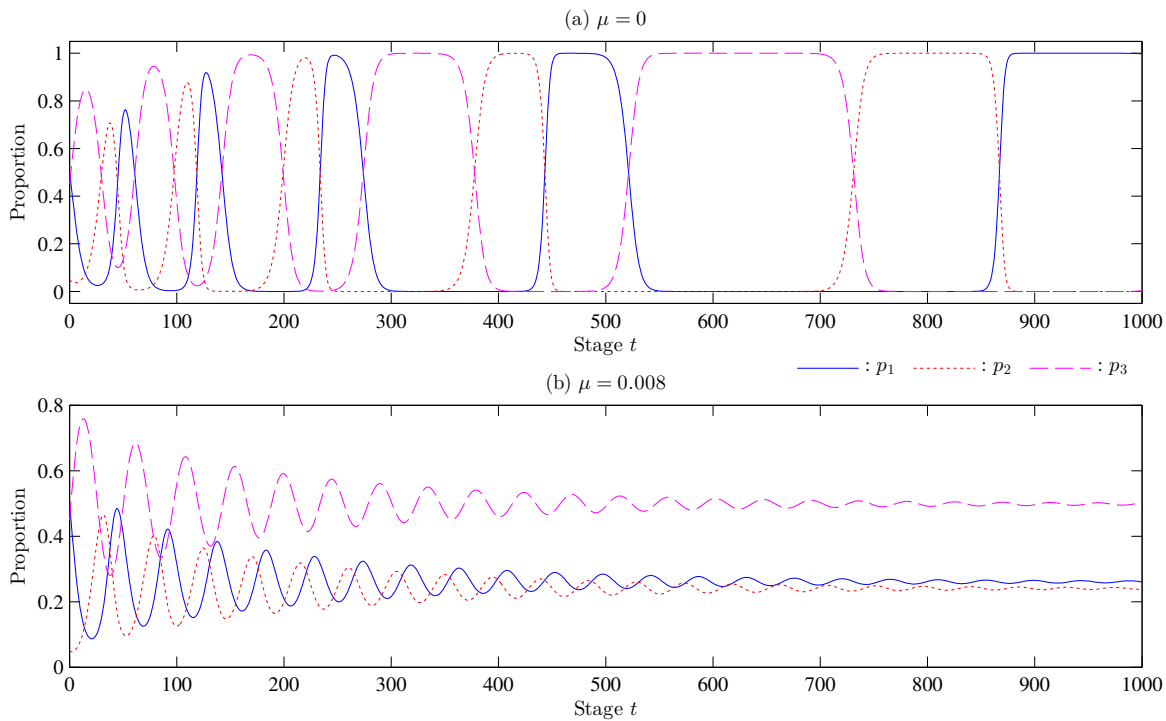}
\caption{The stabilizing effect, through a Hopf bifurcation, of mutations in the case of no dominant type. Evolution of the components of the population state over time for $\mu = 0$ and $\mu = 0.008 > \mu^*$, for the case $\tilde{u}_{12} = 0.7$, $\tilde{u}_{13} = 0.4$, $\tilde{u}_{23} = 0.6$.\label{fig4}}
\end{figure*}

Our numerical investigation reveals that, for a given payoff matrix $\tilde{\mathbf{U}}$, for very low mutation rates, the behavior is, as expected, that trajectories starting near the fixed point spiral outward to a limit cycle that is contained in the interior of $\Delta^3$, while trajectories starting outside the limit cycle spiral inward towards it. As we increase $\mu$ slightly, the limit cycle shrinks, though the behavior of the system does not change qualitatively. As we continue to increase $\mu$, at some point the limit cycle shrinks to the fixed point. As we increase $\mu$ past this value, the fixed point becomes the limit of inward spiraling trajectories. (Increasing $\mu$ still further -- for most payoff matrices, to very large values -- the fixed point eventually becomes regularly stable. Because we typically consider low mutation rates, and because the stability of the fixed point does not change, this result is not of significant interest.)

This is diagnosed by the observation that the magnitude of complex eigenvalues of the Jacobian matrix, evaluated at the fixed point, decreases through unity as we increase $\mu$ through its critical value.

Thus, the system undergoes a supercritical Hopf bifurcation \cite{Strogatz:1994}: as we increase $\mu$ through some critical value $\mu^*$, the fixed point changes from an unstable spiral to a stable spiral. Moreover, for most payoff matrices, this occurs for relatively low values of $\mu$. For example, the critical value of $\mu$ in the case $\tilde{u}_{12}=0.7$, $\tilde{u}_{13}=0.4$, $\tilde{u}_{23}=0.6$ is $\mu^{*} = 0.0057$. We shall see that, fixing $\tilde{u}_{13}$, for small values of $\tilde{u}_{12}$ and $\tilde{u}_{23}$ (i.e.\ close to $\tfrac{1}{2}$), the critical value of $\mu$ is very small, while for large values of $\tilde{u}_{12}$ and $\tilde{u}_{23}$ (i.e.\ close to $1$), the critical value of $\mu$ is larger.

The significance of this result is clear: even very low mutation rates can stabilize a system that is, in the absence of mutations, unstable. The population state in the case of no mutations eventually oscillates wildly between near-pure states, but if we introduce even a very low mutation rate, the population state instead eventually settles to a stable point of fixed proportions. Fig.\ \ref{fig4} demonstrates these different behaviors, plotting the time evolution of the population proportions for the case of a zero mutation rate and the case of a mutation rate slightly above the critical value. Fig.\ \ref{fig5} illustrates the Hopf bifurcation in the phase plane for the same payoff matrix.

\begin{figure*}[t]
\centering
\includegraphics[width=155mm]{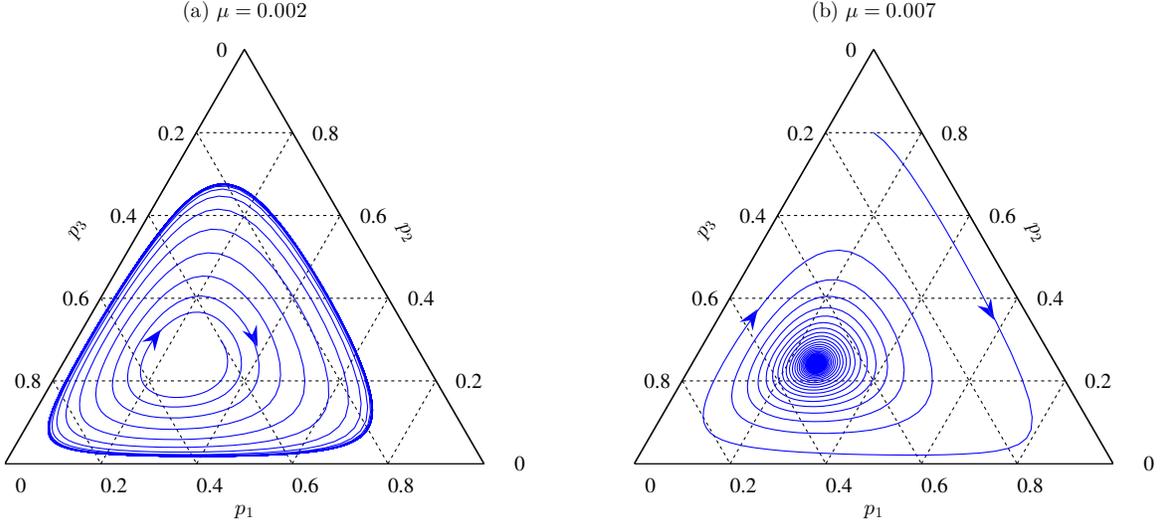}	
\caption{The stabilizing effect, through a Hopf bifurcation, of mutations in the case of no dominant type. Trajectories for (a) $\mu = 0.002<\mu^*$ and (b) $\mu = 0.007>\mu^*$, in both cases for $\tilde{u}_{12} = 0.70$, $\tilde{u}_{13} = 0.40$, $\tilde{u}_{23} = 0.60$. \label{fig5}}
\end{figure*}

Fixing $\tilde{u}_{13} = 0.4$, and for various mutation rates, we plot in Fig.\ \ref{fig6} the regions in the $\tilde{u}_{12}$-$\tilde{u}_{23}$ plane for which the types of behavior described above arise. The region for which the original (zero mutation rate) behavior, that of an unstable spiral, occurs, shrinks as we increase $\mu$. The boundary of the region within each plot corresponds to the payoff matrices for which the value of $\mu$ in question is the critical value. Notice that the payoff matrices with the largest critical mutation rates are those with the largest values for $\tilde{u}_{12}$ and $\tilde{u}_{23}$.

\begin{figure*}[t!]
\centering
\includegraphics[width=155mm]{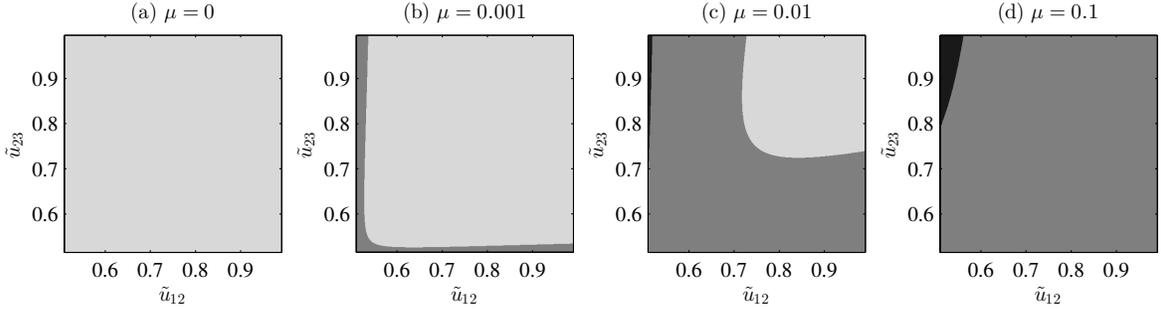}		
\caption{The case of no dominant type. Regions in the $\tilde{u}_{12}$-$\tilde{u}_{23}$ plane that result in fixed points of the various types, for various mutation rates, fixing $\tilde{u}_{13} = 0.4$. The lightest regions correspond to unstable spirals; the mid-tone regions correspond to stable spirals; and the darkest regions correspond to regularly stable fixed points.\label{fig6}}
\end{figure*}
\section{Conclusions} \label{sec:conc}
\subsection{Review}
We have presented a study of a purely competitive dynamics for $3$-type games, arguing that such a dynamics represents a characterization of how certain systems, both natural and artificial, are governed.

Throughout, we maintained the assumptions of an infinite population, and that the probability of advancement of a certain type in an engagement with another is proportional to that type's payoff in the engagement. We further assumed that, in each period, each agent is selected for pairing, though the general model developed in Sec.\ \ref{sec:model} allows for more general selection mechanisms, and indeed for all these assumption to be relaxed.

Initially ignoring the possibility of mutations, we showed that two types of behavior are possible for the population state's evolution, corresponding to two different types of payoff matrix. The first, where one of the types is dominant (in the sense that in an engagement with another type, it always receives the higher payoff) the population state eventually comprises only agents of the dominant type (unless it initially comprised only agents of another type). If no type is dominant, on the other hand, the components of the population state (assuming it is not initially at the mixed-type fixed point, the existence of which is guaranteed for payoff matrices of this type) eventually fluctuate, in procession, between being very close to zero and very close to unity.

Allowing for a nonzero mutation rate, our analysis became predominantly numerical. We showed that, in the case of a dominant type, the behavior of the system does not differ significantly from that under the assumption of a zero mutation rate. In particular, the population state still tends to a globally attracting fixed point near dominant type's pure state. For the case of no dominant type, the dynamics of the system does, in some cases, alter significantly. For very low mutation rates, as expected, the system behaves as it does under the assumption of no mutations; that is, the population state, if starting near the fixed point in the interior, spirals outwards towards a limiting cycle, so that, again, the proportions of the types eventually fluctuate, in procession, with significant amplitude. However, as we increase the mutation rate for a given setup, at some point, the interior fixed point, previously an outward spiral, becomes an inward spiral: the system undergoes a Hopf bifurcation at some critical value of $\mu$. The result is that, for mutation rates larger than the critical value, the fixed point is stable; the population proportions eventually settle down to fixed values, rather than continuing to fluctuate, as in the case of zero or very low mutation rates. Moreover, this critical value was shown to be relatively low for most setups, so that this effect should not be thought of as pathological. That even a very low mutation rate can alter the behavior of the system so significantly -- indeed, entirely reversing its instability -- is a result of both theoretical and practical interest, and it would certainly be worth exploring the significance of this result in the context of real-world dynamical systems.
\subsection{Possibilities for future study}
In considering the possibilities for extension of the analysis carried out here, we naturally turn to the simplifying assumptions we have made, and ask whether the relaxation of any of them could yield further insight into these dynamics.

First, we assumed throughout that $\theta=1$. Increasing the value of this parameter serves to favor the stronger type in a pairing, in the sense that the probability that the stronger type propagates its type to its offspring increases as we increase $\theta$. Of interest is the limiting case \hbox{$\theta=\infty$}, where in a given pairing, the offspring receive with certainty the type of the stronger parent. Intuition suggests that, in the case of a dominant type, that type's predominance in terms of its proportion in the long term population state would be increased relative to the case where $\theta=1$ (for nonzero mutation rates, of course -- for a zero mutation rate, we would expect the long term outcome to be that type's pure state, regardless of $\theta$). In the case of no dominant type, it is not clear whether the dynamics observed for $\theta=1$ would be preserved for $\theta=\infty$, though a numerical investigation would be revealing in this regard.

Second, our assumption that in each period, each agent is selected for mating ($\phi_i \equiv 1$), can be weakened to one where an agent's selection for mating is dependent on its fitness. In this sense, the dynamics would no longer be purely competitive, since an agent's fitness, or average payoff, depends in part on its payoff in engagements with agents of its own type. This would align the dynamics more with the replicator dynamics \cite{hofbauer2003}; whether the resultant behavior would more closely resemble that for the dynamics we considered here or the replicator dynamics could be investigated numerically.

Finally, the assumption of an infinite population was maintained throughout, allowing probabilities and realized proportions to be equated. Of interest would be an investigation of the effect of a finite population on the results above; we might expect the behavior of very large populations to resemble that of an infinite population, but would small populations also exhibit such behavior? Indeed, a preliminary analysis carried out by the authors, making use of simulations of finite population cases, suggests that for some setups, the long run behavior of the population state is qualitatively different for small populations (say, $N=20$) than larger ones (say, $N=200$). A more detailed investigation of the finite population case will be carried out in later work.
\begin{acknowledgments}
The authors thank Martin Wittenberg and Drew Fudenberg for helpful comments, and acknowledge with gratitude the University of Cape Town and the National Research Foundation for the provision of financial support.
\end{acknowledgments}
%
\end{document}